\documentclass[11pt, letter]{article}
\usepackage{amsthm}
\usepackage{graphicx} 

\usepackage{amsmath, amssymb, amsfonts, verbatim}
\usepackage{mathtools}
\usepackage{mathrsfs}
\usepackage{stmaryrd}
\usepackage{tcolorbox}

\usepackage{array}
\usepackage{multirow}

\usepackage{fullpage}

\usepackage{url}

\usepackage{thmtools}

\usepackage[margin=1in]{geometry}
\usepackage{fullpage}
\usepackage{tgtermes}
\usepackage[T1]{fontenc}
\usepackage[colorlinks,citecolor=blue,linkcolor=blue,urlcolor=black]{hyperref}
\usepackage{graphicx}
\usepackage{mathtools}
\usepackage{amsfonts,amsmath,amsthm,amssymb,dsfont}
\usepackage{xfrac,nicefrac}
\usepackage{mathdots}
\usepackage{bm,bbm}
\usepackage{url}
\usepackage{paralist}
\usepackage{enumerate}
\usepackage[normalem]{ulem}
\usepackage{xspace}
\xspaceaddexceptions{]\}}
\usepackage{comment}
\usepackage{tabu}
\usepackage{framed}
\usepackage{float,wrapfig}
\usepackage{tikz}
\usepackage[usenames,dvipsnames]{pstricks} 

\usepackage{enumitem}
\usepackage{adjustbox}
\usepackage{algorithm}
\usepackage[noend]{algpseudocode}
\usepackage[capitalize]{cleveref}
\usepackage{natbib}

\usepackage{thmtools}
\usepackage{thm-restate}

\usetikzlibrary{hobby}
\usetikzlibrary{decorations.markings}
\usetikzlibrary{quotes,angles}

\DeclarePairedDelimiter\abs{\lvert}{\rvert}%
\DeclarePairedDelimiter\norm{\lVert}{\rVert}%

\makeatletter
\let\oldabs\abs
\def\abs{\@ifstar{\oldabs}{\oldabs*}}
\let\oldnorm\norm
\def\norm{\@ifstar{\oldnorm}{\oldnorm*}}
\makeatother

\theoremstyle{plain}

\newtheorem{theorem}{Theorem}[section]
\newtheorem{lemma}[theorem]{Lemma}

\newtheorem{assumption}[theorem]{Assumption}
\theoremstyle{definition}

\newtheorem{remark}[theorem]{Remark}

\def\ShowAuthNotes{1}
\ifnum\ShowAuthNotes=1
\newcommand{\authnote}[2]{\ \\ \textcolor{red}{\parbox{0.9\linewidth}{[{\footnotesize {\bf #1:} { {#2}}}]}}\newline}
\else
\newcommand{\authnote}[2]{}
\fi

\renewcommand{\epsilon}{\varepsilon}

\renewcommand{\tilde}{\widetilde}
\renewcommand{\hat}{\widehat}

\newcommand{\extd}[1]{\hat{d}[#1]}
\newcommand{\exttrued}[1]{d(#1)}

\newcommand{\expectB}{B}
\newcommand{\expectU}{\tilde{U}}

\newcommand{\actualB}{B'}
\newcommand{\actualU}{U}

\newcommand{\pivotU}{W}
\newcommand{\pivotP}{P}

\newcommand{\layer}{l}

\newcommand{\ds}{\mathcal{D}}
\newcommand{\dsiter}[1]{\mathcal{D}_{#1}}

\newcommand{\edgeweight}[2]{w_{#1#2}}

\newcommand{\findpivotconst}{C}
\newcommand{\partitionconst}{C}
\newcommand{\timeconst}{C}

\newcommand{\keyvaluepair}[2]{\langle#1, #2\rangle}

\newcommand{\rawtree}{T}
\newcommand{\tree}[1]{\rawtree(#1)}
\newcommand{\boundtree}[2]{\rawtree_{<#2}(#1)}

\newcommand{\rangetree}[2]{\rawtree_{#2}(#1)}
\newcommand{\htree}{\mathcal{T}}

\newcommand{\xiao}[1]{{\textcolor{purple}{[\textbf{Xiao:} #1]}}}

\title{Breaking the Sorting Barrier for Directed Single-Source Shortest Paths}

\author{
    Ran Duan \thanks{Tsinghua University. Email: duanran@mail.tsinghua.edu.cn, mjy22@mails.tsinghua.edu.cn, ylh21@mails.tsinghua.edu.cn.}
    \and Jiayi Mao \footnotemark[1]
    \and Xiao Mao \thanks{Stanford University. Email: matthew99a@gmail.com.}
    \and Xinkai Shu \thanks{Max Planck Institute for Informatics. Email: xshu@mpi-inf.mpg.de.}
    \and Longhui Yin \footnotemark[1]
}

\begin{document}
\maketitle

\begin{abstract}
    We give a deterministic $O(m\log^{2/3}n)$-time algorithm for single-source shortest paths (SSSP) on directed graphs with real non-negative edge weights in the comparison-addition model. This is the first result to break the $O(m+n\log n)$ time bound of Dijkstra's algorithm on sparse graphs, showing that Dijkstra's algorithm is not optimal for SSSP. 
\end{abstract}


\section{Introduction}

In an $m$-edge $n$-vertex directed graph $G = (V, E)$ with a non-negative weight function $w: E \to \mathbb{R}_{\geq 0}$, single-source shortest path (SSSP) considers the lengths of the shortest paths from a source vertex $s$ to all $v\in V$. Designing faster algorithms for SSSP is one of the most fundamental problems in graph theory, with exciting improvements since the 50s.

The textbook Dijkstra's algorithm \cite{Dij59}, combined with advanced data structures such as the Fibonacci heap \cite{FT87} or the relaxed heap \cite{DJG88}, solves SSSP in $O(m+n\log n)$ time. It works in the \emph{comparison-addition} model, natural for real-number inputs, where only comparison and addition operations on edge weights are allowed, and each operation consumes unit time.
For undirected graphs, Pettie and Ramachandran \cite{PR05} proposed a hierarchy-based algorithm which runs in $O(m\alpha(m,n)+\min\{n\log n,n\log\log r\})$ time in the comparison-addition model, where $\alpha$ is the inverse-Ackermann function and $r$ bounds the ratio between any two edge weights. 

Dijkstra's algorithm also produces an ordering of vertices by distances from the source as a byproduct. A recent contribution by Haeupler, Hlad\'{\i}k, Rozho\v{n}, Tarjan and T\v{e}tek \cite{HHRTT24} showed that Dijkstra's algorithm is optimal if we require the algorithm to output the order of vertices by distances.
If only the distances and not the ordering are required, a recent result by Duan, Mao, Shu and Yin \cite{DMSY23} provided an $O(m\sqrt{\log n\log \log n})$-time randomized SSSP algorithm for undirected graphs, better than $O(n\log n)$ in sparse graphs. However it remains to break such a sorting barrier in directed graphs.


\subsection{Our results}

In this paper, we present the first SSSP algorithm for directed real-weighted graphs that breaks the sorting bottleneck on sparse graphs. 

\begin{theorem}\label{thm:main}
There exists a deterministic algorithm that takes $O(m\log^{2/3}(n))$ time to solve the single-source shortest path problem on directed graphs with real non-negative edge weights.
\end{theorem}

Note that the algorithm in~\cite{DMSY23} is randomized, thus our result is also the first deterministic algorithm to break such $O(m+n\log n)$ time bound even in undirected graphs.


\paragraph{Technical Overview.}
Broadly speaking, there are two traditional algorithms for solving the single-source shortest path problem:
\begin{itemize}
     \item Dijkstra's algorithm~\cite{Dij59}: via a priority queue, it each time extracts a vertex $u$ with the minimum distance from the source, and from $u$ relaxes its outgoing edges. It typically sorts vertices by their distances from the source, resulting in a time complexity of at least $\Theta(n \log n)$.
     \item Bellman-Ford algorithm~\cite{Bellman1958}: based on dynamic programming, it relaxes all edges for several steps. For finding shortest paths with at most $k$ edges, the Bellman-Ford algorithm can achieve this in $O(mk)$ time without requiring sorting.
\end{itemize}
Our approach merges these two methods through a recursive partitioning technique, similar to those used in bottleneck path algorithms as described in~\cite{GT88, CKTZZ, DLWX}.

At any point during the execution of the Dijkstra's algorithm, the priority queue (heap) maintains a ``frontier'' $S$ of vertices such that if a vertex $u$ is ``incomplete'' --- if the current distance estimate $\hat{d}[u]$ is still greater than the true distance $d(u)$ --- the shortest $s$-$u$ path must visit some complete vertex $v\in S$. In this case, we say $u$ is ``dependent'' on $v\in S$. (However, vertices in $S$ are not guaranteed to be all complete.) The Dijkstra's algorithm simply picks the vertex in $S$ closest to source, which must be complete, and then relaxes all edges from that vertex.

The running time bottleneck comes from the fact that sometimes the frontier may contain $\Theta(n)$ vertices. Since we constantly need to pick the vertex closest to source, we essentially need to maintain a total order between a large number of vertices, and are thus unable to break the $\Omega(n \log n)$ sorting barrier. Our most essential idea is a way to reduce the size of the frontier. Suppose we want to compute all distances that are less than some upper bound $B$. Let $\expectU$ be the set of vertices $u$ with $d(u) < B$ and the shortest $s$-$u$ path visits a vertex in $S$. It is possible to limit the size of our frontier $|S|$ to $|\expectU| / \log ^ {\Omega(1)}(n)$, or $1 / \log ^ {\Omega(1)}(n)$ of the vertices of interest. Given a parameter $k = \log ^ {\Omega(1)}(n)$, there are two possible cases:
\begin{itemize}
    \item If $|\expectU| > k|S|$, then our frontier already has size $|\expectU| / k$;
    \item Otherwise, suppose $|\expectU| \le k|S|$. By running Bellman-Ford step $k$ times from vertices in $S$, every vertex $u\in \expectU$ whose shortest $s$-$u$ path containing $<k$ vertices in $\expectU$ is complete. Otherwise the vertex $v\in S$ which $u$ is dependent on must have a shortest path tree rooted at it with $\geq k$ vertices in $\expectU$, so we can shrink the frontier $S$ to the set of ``pivots'', each of which has a shortest path tree of size $\geq k$, and the number of such pivots is bounded by $|\expectU|/k$.    
\end{itemize}

Our algorithm is based on the idea above, but instead of using a Dijkstra-like method where the frontier is dynamic and thus intractable, we use a divide-and-conquer procedure that consists of $\log n / t$ levels, each containing a set of frontier vertices and an upper bound $B$, such that a naive implementation would still spend $\Theta(t)$ time per frontier vertex and the running time would remain $\Theta(\log n)$ per vertex. 
We can however apply the aforementioned frontier reduction on each level so that the $\Theta(t)$ work only applies to the pivots, about $1 / \log ^ {\Omega(1)}(n)$ of the frontier vertices. Thus the running time per vertex gets reduced to $\log n / \log ^ {\Omega(1)}(n)$, which is a significant speedup.

\subsection{Related Works}
For the SSSP problem, if we allow the algorithm to run in the word RAM model with integer edge weights, a sequence of improvements beyond $O(n\log n)$ \cite{FW93,FW94, Thorup96, Raman96, Raman97, TM00, HT20} culminated in Thorup's linear-time algorithm for undirected graphs~\cite{Thorup00} and $O(m+n\log\log\min\{n,C\})$ time for directed graphs~\cite{Thorup04}, where $C$ is the maximum edge weight. 
For graphs with negative weights, recent breakthrough results include near-linear time $O(m^{1+o(1)}\log C)$ algorithms for SSSP with negative integer weights~\cite{flow,BNW22, BCF23}, and strongly subcubic time algorithm for SSSP with negative real weights~\cite{Fineman24,HJQ24}.

Based on the lower bound of $\Omega(n \log n)$ for comparison-based sorting algorithms, it was generally believed that such \emph{sorting barrier} exists for SSSP and many similar problems. Researchers have broken such a sorting barrier for many graph problems.
For example, Yao \cite{Yao75} gave a minimum spanning tree (MST) algorithm with running time $O(m\log\log n)$, which is further improved to randomized linear time \cite{KKT95}.
Gabow and Tarjan \cite{GT88} solves $s$-$t$ bottleneck path problem in $O(m\log^* n)$-time, later improved to randomized $O(m\beta(m,n))$ time~\cite{CKTZZ}, where $\beta(m,n)=\min\{k\geq 1: \log^{(k)}n\leq\frac{m}{n}\}$. For the single-source all-destination bottleneck path problem, there is a randomized $O(m\sqrt{\log n})$-time algorithm by Duan, Lyu, Wu and Xie \cite{DLWX}. For the single-source nondecreasing path problem, Virginia V.Williams~\cite{VW10} proposed an algorithm with a time bound of $O(m\log\log n)$.


\section{Preliminaries}
\label{sec:preliminaries}

We consider a directed graph $G=(V, E)$ with a non-negative weight function $w: E \to \mathbb{R}_{\geq 0}$, also denoted by $\edgeweight{u}{v}$ for each edge $(u,v)\in E$. Let $n = \abs{V}$, $m=\abs{E}$ be the number of vertices and edges in the graph.
In the single-source shortest path problem, we assume there is a source vertex in the graph denoted by $s$. The goal of our algorithm is to find the length of shortest path from $s$ to every vertex $v\in V$, denoted by $\exttrued{v}$. Without loss of generality, we assume that every vertex in $G$ is reachable from $s$, so $m\geq n-1$.

\paragraph{Constant-Degree Graph.} In this paper we assume that the algorithm works on a graph with constant in-degrees and out-degrees. For any given graph $G$ we may construct $G'$ by a classical transformation (similar to~\cite{Frederickson83}) to accomplish this:
\begin{itemize}
    \item Substitute each vertex $v$ with a cycle of vertices strongly connected with zero-weight edges. For every incoming or outgoing neighbor $w$ of $v$, there is a vertex $x_{vw}$ on this cycle.
    \item For every edge $(u,v)$ in $G$, add a directed edge from vertex $x_{uv}$ to $x_{vu}$ with weight $\edgeweight{u}{v}$.
\end{itemize}
It's clear that the shortest path is preserved by the transformation. Each vertex in $G'$ has in-degree and out-degree at most 2, while $G'$ being a graph with $O(m)$ vertices and $O(m)$ edges.

\paragraph{Comparison-Addition Model.} Our algorithm works under the \emph{comparison-addition} model, where all edge weights are subject to only comparison and addition operations. In this model, each addition and comparison takes unit time, and no other computations on edge weights are allowed.

\paragraph{Labels Used in the Algorithm.}
For a vertex $v\in V$, denote $\exttrued{u}$ as the length of shortest path from $s$ to $u$ in the graph. Similar to the Dijkstra's algorithm, our algorithm maintains a global variable $\extd{u}$ as a sound estimation of $\exttrued{u}$ (that is, $\extd{u} \geq \exttrued{u}$). Initially $\extd{s}=0$ and $\extd{v}=\infty$ for any other $v\in V$.
Throughout the algorithm we only update $\extd{v}$ in a non-increasing manner by relaxing an edge $(u, v)\in E$, that is, assign $\extd{v}\gets \extd{u}+\edgeweight{u}{v}$ when $\extd{u}+\edgeweight{u}{v}$ is no greater than the old value of $\extd{v}$.
Therefore each possible value of $\extd{v}$ corresponds to a path from $s$ to $v$.
If $\extd{x} = \exttrued{x}$, we say $x$ is \textit{complete}; otherwise, we say $x$ is \textit{incomplete}. If all vertices in a set $S$ are complete, we say $S$ is complete. Note that completeness is sensitive to the progress of the algorithm.
The algorithm also maintains a shortest path tree according to current $\extd{\cdot}$ by recording the predecessor $\textsc{Pred}[v]$ of each vertex $v$. When relaxing an edge $(u,v)$, we set $\textsc{Pred}[v] \gets u$.

\paragraph{Total order of Paths.}
Like in many papers on shortest path algorithms, we make the following assumption for an easier presentation of our algorithm:
\begin{assumption}\label{assumption:unique-path}
    All paths we obtain in this algorithm have different lengths.
\end{assumption}
This assumption is required for two purposes:
\begin{enumerate}
    \item To ensure that $\textsc{Pred}[v]$ for all vertices $v\in V$ keep a tree structure throughout the algorithm;
    \item To provide a relative order of vertices with the same value of $\extd{}$.
\end{enumerate}
Next, we show that this assumption does not lose generality since we can provide a total order for all paths we obtain. We treat each path of length $l$ that traverses $\alpha$ vertices $v_1=s, v_2, ..., v_\alpha$ as a tuple $\langle l, \alpha, v_\alpha, v_{\alpha-1},...,v_1\rangle$ (note that the sequence of vertices is reversed in the tuple). We sort paths in the lexicographical order of the tuple in principle. That is, first compare the length $l$. When we have a tie, compare $\alpha$. If there is still a tie, compare the sequence of vertices from $v_\alpha$ to $v_1$. 
Comparison between the tuples can be done in only $O(1)$ time with extra information of $\textsc{Pred}[]$ and $\alpha$ stored for each $\extd{v}$:
\begin{description}
    \item[Relaxing an edge $(u,v)$:] If $u\neq\textsc{Pred}[v]$, even if there is a tie in $l$ and $\alpha$, it suffices to compare between $u$ and $\textsc{Pred}[v]$, and if $u=\textsc{Pred}[v]$, then $\extd{u}$ is updated to currently ``shortest'' and $\extd{v}$ needs to get updated;
    \item [Comparing two different $\extd{u}$ and $\extd{v}$ for $u \neq v$:] Even if there is a tie in $l$ and $\alpha$, it suffices to compare the endpoints $u$ and $v$ only.
\end{description}
Therefore, in the following sections, we may assume that each path started from $s$ has a unique length.





\section{The Main Result}


\subsection{The Algorithm}

Recall that we work on SSSP from source $s$ in constant-degree graphs with $m=O(n)$. Let $k := \lfloor \log ^ {1/3}(n) \rfloor$ and $t := \lfloor \log ^ {2/3}(n) \rfloor$ be two parameters in our algorithm.
Our main idea is based on divide-and-conquer on vertex sets. We hope to partition a vertex set $U$ into $2^t$ pieces $U = U_1\cup U_2\cdots \cup U_{2^t}$ of similar sizes, where vertices in earlier pieces have smaller distances, and then recursively partition each $U_i$. In this way, the size of the sub-problem shrinks to a single vertex after roughly $(\log n) / t$ recursion levels. 
To construct our structure dynamically, each time we would try to compute the distances to a set of closest vertices (without necessarily recovering a complete ordering between their true distances), and report a boundary indicating how much progress we actually make.

Suppose at some stage of the algorithm, for every $u$ with $\exttrued{u}< b$, $u$ is complete and we have relaxed all the edges from $u$. We want to find the true distances to vertices $v$ with $\exttrued{v}\geq b$. To avoid the $\Theta(\log n)$ time per vertex in a priority queue, consider the ``frontier'' $S$ containing all current vertices $v$ with $b\leq \extd{v}<B$ for some bound $B$ (without sorting them). We can see that the shortest path to every incomplete vertex $v'$ with $b\leq \exttrued{v'}<B$ must visit some complete vertex in $S$. Thus to compute the true distance to every $v'$ with $b\leq \exttrued{v'}<B$, it suffices to find the shortest paths from vertices in $S$ and bounded by $B$.
We call this subproblem \emph{bounded multi-source shortest path} (BMSSP) and design an efficient algorithm to solve it. The following lemma summarizes what our BMSSP algorithm achieves. 



\begin{lemma} [Bounded Multi-Source Shortest Path] \label{lemma:bmssp}
We are given an integer level $\layer \in [0, \lceil \log (n) / t\rceil ]$,  a set of vertices $S$ of size $\leq 2 ^ {\layer t}$, and an upper bound $\expectB > \max_{x\in S}\extd{x}$.
Suppose that for every incomplete vertex $v$ with $\exttrued{v} < \expectB$, the shortest path to $v$ visits some complete vertex $u\in S$.

Then we have an sub-routine $\textsc{BMSSP}(\layer, \expectB, S)$ (\cref{alg:main}) in $O((kl + tl/k + t)|U|)$ time that outputs a new boundary $\actualB \leq \expectB$ and a vertex set $\actualU$ that contains every vertex $v$ with $\exttrued{v} < \actualB$ and the shortest path to $v$ visits some vertex of $S$. At the end of the sub-routine, $\actualU$ is complete. Moreover, one of the following is true:
\begin{description}
    \item[Successful execution] $\actualB = \expectB$.
    \item[Partial execution due to large workload] $\actualB < \expectB$, and $|\actualU| = \Theta\left(k 2 ^ {\layer t}\right)$.
\end{description}

\end{lemma}

On the top level of divide and conquer, the main algorithm calls \textsc{BMSSP} with parameters $l = \lceil (\log n)/t\rceil$, $S = \{s\}, B = \infty$. Because $|\actualU| \leq |V| = o(kn)$, it must be a successful execution, and the shortest paths to all vertices are found. With chosen $k$ and $t$, the total running time is $O(m\log^{2/3}n)$.

\textsc{BMSSP} procedure on level $l$ works by recursion, it first tries to ``shrink'' $S$ to size $|U|/k$ by a simple Bellman-Ford-like method (\cref{lemma:findpivots}), then it makes several recursive calls on level $(\layer - 1)$ until the bound reaches $B$ or the size of $U$ reaches $\Theta\left(k 2 ^ {\layer t}\right)$. We always make sure that a recursive call solves a problem that is smaller by a factor of roughly $1 / 2 ^ t$, so the number of levels of the recursion is bounded by $O((\log n) / t)$.
The main ideas are:

\begin{itemize}
    \item Every time we only select about $2^{(\layer - 1)t}$ vertices for our next recursive call, so we use a partial sorting heap-like data structure as described in~\cref{lemma:partition} to improve the running time.
    \item If we used all of $S$ in the recursive calls, then after all remaining levels $S$ could be fully sorted and nothing is improved. Thus \textsc{FindPivots} (\cref{alg:find-pivots}) procedure is crucial here, as it shows that only at most $|U|/k$ vertices of $S$ are useful in the recursive calls.
    \item Partial executions may be more complicated to analyze,
    so we introduce some techniques like \textsc{BatchPrepend} operation in \cref{lemma:partition}.
\end{itemize}



\paragraph{Finding Pivots.}


Recall that in the current stage, every $u$ such that $\exttrued{u}< b$ is complete, and we have relaxed all the edges from such $u$'s, and the set $S$ includes every vertex $v$ with current $b\leq\extd{v}<B$. Thus, the shortest path of every incomplete vertex $v$ such that $d(v) < B$ visits some complete vertex in $S$. (This is because the first vertex $w$ in the shortest path to $v$ with $\exttrued{w}\geq b$ is complete and included in $S$.)

The idea of finding pivots is as follows: we perform relaxation for $k$ steps (with a bound $B$). After this, if the shortest $s$-$v$ path with $b\leq\exttrued{v}<B$ goes through at most $k$ vertices $w$ with $b\leq\exttrued{w}<B$, then $v$ is already complete. Observe that the number of large shortest path trees from $S$, consisting of at least $k$ vertices and rooted at vertices in $S$, is at most $|\expectU| / k$ here, where $\expectU$ is the set of all vertices $v$ such that $\exttrued{v} < \expectB$ and the shortest path of $v$ visits some complete vertex in $S$. So only the roots of such shortest-path trees are needed to be considered in the recursive calls, and they are called ``pivots''.







\begin{lemma} [Finding Pivots] \label{lemma:findpivots}
Suppose we are given a bound $\expectB$ and a set of vertices $S$.
Suppose that for every incomplete vertex $v$ such that $\exttrued{v} < \expectB$, the shortest path to $v$ visits some complete vertex $u\in S$.

Denote by $\expectU$ the set that contains every vertex $v$ such that $\exttrued{v} < \expectB$ and the shortest path to $v$ passes through some vertex in $S$. The sub-routine $\textsc{FindPivots}(\expectB, S)$ (\cref{alg:find-pivots}) finds a set $\pivotU \subseteq \expectU$ of size   $O(k|S|)$ and a set of pivots $\pivotP \subseteq S$ of size at most $|\pivotU| / k$ such that for every vertex $x \in \expectU$, at least one of the following two conditions holds:
\begin{itemize}
    \item At the end of the sub-routine, $x \in \pivotU$ and $x$ is complete;
    \item The shortest path to $x$ visits some complete vertex $y \in \pivotP$.
\end{itemize}
Moreover, the sub-routine runs in time $O(k|\pivotU|)=O(\min\{k^2|S|, k|\expectU|\})$.
\end{lemma}

\begin{algorithm}
    \caption{Finding Pivots}
    \label{alg:find-pivots}
    \begin{algorithmic}[1]
        \Function{\textsc{FindPivots}}{$B, S$}
            \Statex \textbullet~\textbf{requirement:} for every incomplete vertex $v$ with $\exttrued{v} < \expectB$, the shortest path to $v$ visits some complete vertex in $S$
            \Statex \textbullet~\textbf{returns:} sets $\pivotP, \pivotU$ satisfying the conditions in \cref{lemma:findpivots}
            \State $W\gets S$
            \State $W_0\gets S$
            \For {$i \gets 1$ to $k$} \Comment{Relax for $k$ steps}
                \State $W_i \gets \emptyset$
                \For {all edges $(u,v)$ with $u\in W_{i-1}$}
                    \If {$\extd{u} + \edgeweight{u}{v} \leq \extd{v}$}\label{line:relaxfootnote} 
                        \State $\extd{v} \gets \extd{u} + \edgeweight{u}{v}$
                        \If {$\extd{u} + \edgeweight{u}{v} < B$}
                            \State $W_i\gets W_i \cup \{v\}$
                        \EndIf
                    \EndIf
                \EndFor
                \State $W \gets W \cup W_i$
                \If {$|W|>k|S|$}
                    \State $P\gets S$
                    \State \Return $P, W$
                \EndIf
            \EndFor
            \State $F\gets \{(u,v)\in E:u,v\in W, \extd{v}=\extd{u} + \edgeweight{u}{v}\}$ \Comment{$F$ is a directed forest under~\cref{assumption:unique-path}}
            \State $P\gets \{u\in S: \text{$u$ is a root of a tree with $\geq k$ vertices in $F$}\}$
            \State \Return $P, W$
            \EndFunction
    \end{algorithmic}
\end{algorithm}

\begin{proof}
Note that all vertices visited by \cref{alg:find-pivots} are in $\expectU$, and those in $W$ are not guaranteed to be complete after the procedure. Also note that every vertex in $\expectU$ must visit some vertex in $S$ which was complete before \cref{alg:find-pivots}, since any incomplete vertex $v$ at that time with $d(v)<B$ must visit some complete vertex in $S$.

If the algorithm returns due to $|W|>k|S|$, we still have $|W| = O(k|S|)$ since the out-degrees of vertices are constant. For every incomplete vertex $v$ such that $\exttrued{v} < \expectB$, we know that the shortest path to $v$ visits some complete vertex $u\in S$. Since $P = S$, we must also have $u \in P$ and conditions in the lemma are thus satisfied.

If $|W|\leq k|S|$, $P$ is derived from $F$.
For any vertex $x\in \expectU$, consider the first vertex $y\in S$ on the shortest path to $x$ which was complete before \cref{alg:find-pivots}, then $y$ must be a root of a tree in $F$. If there are no more than $k-1$ edges from $y$ to $x$ on the path, $x$ is complete and added to $W$ after $k$ relaxations. Otherwise, the tree rooted at $y$ contains at least $k$ vertices, therefore $y$ is added to $P$.
Additionally, $|P|\leq |W|/k \leq |\expectU|/k$ since each vertex in $P$ covers a unique subtree of size at least $k$ in $F$.

To evaluate the time complexity, note that the size of $W$ is $O(\min\{k|S|,|\expectU|\})$, so each iteration takes $O(\min\{k|S|,|\expectU|\})$ time. Computing $P$ after the for-loop runs in $O(|W|)$ time for final $W$. Therefore, \textsc{FindPivots} finishes in $O(\min\{k^2|S|, k|\expectU|\})$ time. 
\end{proof}

We also use the following data structure to help adaptively partition a problem into sub-problems, specified in \cref{lemma:partition}:
\begin{lemma} \label{lemma:partition}
Given at most $N$ key/value pairs to be inserted, an integer parameter $M$, and an upper bound $\expectB$ on all the values involved, there exists a data structure that supports the following operations:

\begin{description}
    \item[Insert] Insert a key/value pair in amortized $ O(\max\{1,\log (N / M)\}) $ time. If the key already exists, update its value.
    
    \item[Batch Prepend] 
    Insert $L$ key/value pairs such that each value in $L$ is smaller than any value currently in the data structure, in amortized $O(L \cdot \max\{1, \log (L / M)\})$ time. If there are multiple pairs with the same key, keep the one with the smallest value. 
    
    \item[Pull] Return a subset $S'$ of keys where $|S'|\le M$ associated with the smallest $|S'|$ values and an upper bound $x$ that separates $S'$ from the remaining values in the data structure, in amortized $O(|S'|)$ time. Specifically, if there are no remaining values, $x$ should be $\expectB$. Otherwise, $x$ should satisfy $\max(S')<x\le \min(D)$ where $D$ is the set of elements in the data structure after the pull operation.

\end{description}
\end{lemma}


\begin{proof}
    We introduce a block-based linked list data structure to support the required operations efficiently. 

Specifically, the data is organized into two sequences of blocks, $\mathcal{D}_0$ and $\mathcal{D}_1$. $\mathcal{D}_0$ only maintains elements from batch prepends while $\mathcal{D}_1$ maintains elements from insert operations, so every single inserted element is always inserted to $\mathcal{D}_1$. Each block is a linked list containing at most $M$ key/value pairs. The number of blocks in $\mathcal{D}_1$ is bounded by $O(\max\{1, N/M\})$, while $\mathcal{D}_0$ does not have such a requirement. Blocks are maintained in the sorted order according to their values, 
that is, for any two key/value pairs $\langle a_1,b_1\rangle \in B_i$ and $\langle a_2, b_2\rangle \in B_j$ , where $B_i$ and $B_j$ are the $i$-th and $j$-th blocks from a block sequence, respectively, and $i<j$, we have $b_1 \le  b_2$. For each block in $\mathcal{D}_1$, we maintain an upper bound for its elements, so that the upper bound for a block is no more than any value in the next block. We adopt a self-balancing binary search tree (e.g.~Red-Black Tree~\cite{red-black}) to dynamically maintain these upper bounds, with $ O(\max\{1,\log (N / M)\})$ search/update time. 

For cases where multiple pairs with the same key are added, we record the status of each key and the associated value. 
If the new pair has a smaller value, we first delete the old pair then insert the new one,
ensuring that only the most advantageous pair is retained for each key.

For the operations required in \cref{lemma:partition}:
\setdescription{font=\normalfont}
\begin{description}
\item[$\textsc{Initialize}(M,B)$]  Initialize $\mathcal{D}_0$ with an empty sequence and $\mathcal{D}_1$ with a single empty block with upper bound $B$. Set the parameter $M$.

\item[$\textsc{Delete}(a,b)$] To delete the key/value pair $\langle a,b\rangle$, we remove it directly from the linked list, which can be done in $O(1)$ time. Note that it's unnecessary to update the upper bounds of blocks after a deletion. However, if a block in $\mathcal{D}_1$ becomes empty after deletion, we need to remove its upper bound in the binary search tree in $O(\max\{1,\log (N / M)\})$ time. Since \textsc{Insert} takes $O(\max\{1,\log (N / M)\})$ time for $\mathcal{D}_1$, deletion time will be amortized to insertion time with no extra cost.

\item[$\textsc{Insert}(a,b)$] To insert a key/value pair $\langle a,b\rangle$, we first check the existence of its key $a$. If $a$ already exists, we delete original pair $\langle a,b'\rangle$ and insert new pair $\langle a,b\rangle$ only when $b<b'$.


Then we insert $\langle a,b\rangle$ to $\mathcal{D}_1$. We first locate the appropriate block for it, which is the block with the smallest upper bound greater than or equal to $b$, using binary search (via the binary search tree) on the block sequence. $\langle a,b\rangle$ is then added to the corresponding linked list in $O(1)$ time. 
Given that the number of blocks in $\mathcal{D}_1$ is $O(\max\{1,N/M\})$, as we will establish later, the total time complexity for a single insertion is $ O(\max\{1,\log (N / M)\})$. 

After an insertion, the size of the block may increase, and if it exceeds the size limit $M$, a split operation will be triggered.

    \item [$\textsc{Split}$] When a block in $\mathcal{D}_1$ exceeds $M$ elements, we perform a split. First, we identify the median element within the block in $O(M)$ time~\cite{median-find}, partitioning the elements into two new blocks each with at most $\lceil M / 2\rceil$ elements --- elements smaller than the median are placed in the first block, while the rest are placed in the second. This split ensures that each new block retains about $\lceil M/2\rceil$ elements while preserving inter-block ordering, so the number of blocks in $\mathcal{D}_1$ is bounded by $ O(N / M) $. (Every block in $\mathcal{D}_1$ contains $\Theta(M)$ elements, including the elements already deleted.)
    After the split, we make the appropriate changes in the binary search tree of upper bounds in $O(\max\{1,\log (N / M)\})$ time.


    \item[$\textsc{BatchPrepend}(\mathcal{L})$] Let $L$ denote the size of $\mathcal{L}$. When $L\le M $, we simply create a new block for $\mathcal{L}$ and add it to the beginning of $\mathcal{D}_0$. Otherwise, we create $O(L/M)$ new blocks in the beginning of $\mathcal{D}_0$, each containing at most $\lceil M/2\rceil $ elements. We can achieve this by repeatedly taking medians which completes in $O(L\log(L/M))$ time.
    



    \item[$\textsc{Pull}()$] 
    
    To retrieve the smallest $M$ values from $\mathcal{D}_0 \cup \mathcal{D}_1$, we collect a sufficient prefix of blocks from $\mathcal{D}_0$ and $\mathcal{D}_1$ separately, denoted as $S'_0$ and $S'_1$, respectively. 
    That is, in $\mathcal{D}_0$ ($\mathcal{D}_1$) we start from the first block and stop collecting as long as we have collected all the remaining elements or the number of collected elements in $S_0'$ ($S_1'$) has reached $M$. If $S_0'\cup S_1'$ contains no more than $M$ elements, it must contain all blocks in $\mathcal{D}_0 \cup \mathcal{D}_1$, so we return all elements in $S_0'\cup S_1'$ as $S'$ and set $x$ to the upper bound $B$, and the time needed is $O(|S'|)$. Otherwise, we want to make $|S'|=M$, and because the block sizes are kept at most $M$, the collecting process takes $O(M)$ time. 

    Now we know the smallest $M$ elements must be contained in $S_0' \cup S_1'$ and can be identified from $S_0' \cup S_1'$ as $S'$ in $O(M)$ time. Then we delete elements in $S'$ from $\mathcal{D}_0$ and $\mathcal{D}_1$, whose running time is amortized to insertion time. Also set returned value $x$ to the smallest remaining value in $\mathcal{D}_0 \cup \mathcal{D}_1$, which can also be found in $O(M)$ time.


    
\end{description}
\end{proof}

We now describe our BMSSP algorithm (\cref{alg:main}) in detail. Recall that the main algorithm calls \textsc{BMSSP} with parameters $l = \lceil (\log n)/t\rceil, S = \{s\}, B = \infty$ on the top level.

In the base case of $\layer =  0$, $S$ is a singleton $\{x\}$ and $x$ is complete. We run a mini Dijkstra's algorithm (\cref{alg:base-case}) starting from $x$ to find the closest vertices $v$ from $x$ such that $\exttrued{v} < \expectB$ and the shortest path to $v$ visits $x$, until we find $k + 1$ such vertices or no more vertices can be found. Let $U_0$ be the set of them. If we do not find $k + 1$ vertices, return $\actualB \gets \expectB, \actualU \gets U_0$. Otherwise, return $\actualB \gets \max_{v\in U_0} \exttrued{v}, \actualU \gets \{v\in U_0:\exttrued{v} < \actualB\}$.

If $\layer > 0$, first we use $\textsc{FindPivots}$ from \cref{lemma:findpivots} to obtain a pivot set $\pivotP\subseteq S$ and a set $\pivotU$.
We initialize the data structure $\ds$ from \cref{lemma:partition} with $M := 2 ^ {(\layer - 1)t}$ and add each $x \in \pivotP$ as a key associated with value $\extd{x}$ to $\ds$.
For simplicity, we write $\ds$ as a set to refer to all the keys (vertices) in it.

Set $\actualB_{0} \gets \min_{x\in \pivotP} \extd{x}$; $\actualU \gets \emptyset$. The rest of the algorithm repeats the following iteration of many phases, during the $i$-th iteration, we:
\begin{enumerate}
    \item Pull from $\ds$ a subset $S_{i}$ of keys associated with the smallest values and $\expectB_i$ indicating a lower bound of remaining value in $\ds$ ;
    \item Recursively call $\textsc{BMSSP}(\layer - 1, \expectB_i, S_{i})$, obtain its return, $\actualB_i$ and $\actualU_i$, and add vertices in $\actualU_{i}$ into $\actualU$;
    \item Relax every edge $(u, v)$ for $u \in \actualU_i$ (i.e.~$\extd{v}\gets \min\{\extd{v}, \extd{u} + \edgeweight{u}{v}\}$).
    If the relax update is valid ($\extd{u} + \edgeweight{u}{v} \leq \extd{v}$), do the following (even when $\extd{v}$ already equals $\extd{u} + \edgeweight{u}{v}$):
    \begin{enumerate}
        \item if $\extd{u} + \edgeweight{u}{v} \in [\expectB_i, \expectB)$, then simply insert $\keyvaluepair{v}{\extd{u} + \edgeweight{u}{v}}$ into $\ds$;
        \item if $\extd{u} + \edgeweight{u}{v} \in [\actualB_i, \expectB_i)$, then record $\keyvaluepair{v}{\extd{u} + \edgeweight{u}{v}}$ in a set $K$;
    \end{enumerate}
    \item Batch prepend all records in $K$ and $\keyvaluepair{x}{\extd{x}}$ for $x\in S_{i}$ with $\extd{x} \in [\actualB_i, \expectB_i)$ into $\ds$;
    \item If $\ds$ is empty, then it is a successful execution and we return;
    \item If $|\actualU| > k2^{\layer t}$, set $\actualB \gets \actualB_{i}$. We expect a large workload and we prematurely end the execution.
\end{enumerate}

Finally, before we end the sub-routine, we update $\actualU$ to include every vertex $x$ in the set $\pivotU$ returned by $\textsc{FindPivots}$ with $\extd{x} < \actualB$.

\begin{algorithm}[H]
    \caption{Base Case of BMSSP}\label{alg:base-case}
    \begin{algorithmic}[1]
        \Function{\textsc{BaseCase}}{$\expectB, S$}
            \Statex \textbullet~\textbf{requirement 1:} $S = \{x\}$ is a singleton, and $x$ is complete 
            \Statex \textbullet~\textbf{requirement 2:} for every incomplete vertex $v$ with $\exttrued{v} < \expectB$, the shortest path to $v$ visits $x$
            \Statex \textbullet~\textbf{returns 1:} a boundary $\actualB \leq \expectB$
            \Statex \textbullet~\textbf{returns 2:} a set $\actualU$
            \State $U_0 \gets S$
            \State initialize a binary heap $\mathcal{H}$ with a single element $\langle x,\extd{x}\rangle$ \cite{williams1964}
            \While{$\mathcal{H}$ is non-empty and $|U_0| < k+1$}
                \State $\langle u,\extd{u}\rangle \gets \mathcal{H}.\textsc{ExtractMin}()$
                \State $U_0 \gets U_0 \cup \{u\}$
                \For {edge $e = (u, v)$}
                    \If {$\extd{u} + \edgeweight{u}{v} \leq \extd{v}$ and $\extd{u} + \edgeweight{u}{v} < \expectB$} \label{line:relax2} 
                        \State $\extd{v} \gets \extd{u} + \edgeweight{u}{v}$  
                        \If{$v$ is not in $\mathcal{H}$}
                            \State $\mathcal{H}.\textsc{Insert}(\langle v,\extd{v}\rangle)$
                        \Else
                            \State $\mathcal{H}.\textsc{DecreaseKey}(\langle v,\extd{v}\rangle)$
                        \EndIf
                    \EndIf
                \EndFor
            \EndWhile
            \If {$|U_0| \le k$}
                \State \Return $\actualB \gets \expectB, \actualU \gets U_0$
            \Else 
                \State \Return $\actualB \gets \max_{v\in U_0} \extd{v}, \actualU \gets \{v\in U_0:\extd{v} < \actualB\}$
            \EndIf
        \EndFunction
    \end{algorithmic}
\end{algorithm}

\begin{algorithm}[ht]
    \caption{Bounded Multi-Source Shortest Path} \label{alg:main}
    \begin{algorithmic}[1]
        \Function{\textsc{BMSSP}}{$\layer, \expectB, S$}
            \Statex \textbullet~\textbf{requirement 1:} $|S| \leq 2^{\layer t}$ 
            \Statex \textbullet~\textbf{requirement 2:} for every incomplete vertex $x$ with $\exttrued{x} < \expectB$, the shortest path to $x$ visits some complete vertex $y \in S$
            
            \Statex \textbullet~\textbf{returns 1:} a boundary $\actualB \leq \expectB$
            \Statex \textbullet~\textbf{returns 2:} a set $\actualU$
            \If {$\layer = 0$}
                \State \Return $\actualB, \actualU \gets \textsc{BaseCase}(B,S)$
            \EndIf
            \State $\pivotP, \pivotU \gets \textsc{FindPivots}(\expectB, S)$ \label{code:wp}
            \State $\ds.\textsc{Initialize}(M,B)$ with $M= 2 ^ {(\layer - 1)t}$ \Comment{ $\ds$ is an instance of \cref{lemma:partition} }
            \State $\ds.\textsc{Insert}(\langle x, \extd{x}\rangle)$ \textbf{for} $x\in \pivotP$ \label{code:ds_insert}
            \State $i\gets 0$; $\actualB_{0} \gets \min_{x\in \pivotP}\extd{x}$; $\actualU \gets \emptyset$ \Comment{If $\pivotP=\emptyset$ set $\actualB_{0} \gets \expectB$}
            \While {$|\actualU| < k2 ^ {\layer t}$ and $\ds$ is non-empty}
                \State $i\gets i+1$
                \State $\expectB_i, S_{i} \gets \ds.\textsc{Pull}()$ \label{code:pull}
                \State $\actualB_i, \actualU_i \gets \textsc{BMSSP}(\layer - 1, \expectB_{i}, S_{i})$ \label{code:recursive}
                \State $\actualU \gets \actualU \cup \actualU_i$
                \State $K \gets \emptyset$
                \For {edge $e = (u, v)$ where $u\in \actualU_{i}$}
                    \If {$\extd{u} + \edgeweight{u}{v} \leq \extd{v}$} \label{code:relax_condition} 
                        \State $\extd{v} \gets \extd{u} + \edgeweight{u}{v}$  
                        \If {$\extd{u} + \edgeweight{u}{v} \in [\expectB_i, \expectB)$}
                            \State $\ds.\textsc{Insert}(\langle v, \extd{u} + \edgeweight{u}{v} \rangle$) \label{code:single_insert}
                        \ElsIf {$\extd{u} + \edgeweight{u}{v} \in [\actualB_i, \expectB_i)$}
                            \State $K \gets K \cup \{\langle v, \extd{u} + \edgeweight{u}{v}\rangle\}$ \label{code:k_insert}
                        \EndIf
                    \EndIf
                \EndFor
                \State $\ds.\textsc{BatchPrepend}(K \cup \{\langle x,\extd{x}\rangle:x\in S_i \text{ and }\extd{x} \in [\actualB_i, \expectB_i)\} )$ \label{code:batch_prepend}
            \EndWhile
            \State \Return $\actualB \gets \min\{\actualB_{i}, \expectB\}$; $\actualU \gets \actualU \cup \{x\in \pivotU: \extd{x} < \actualB\}$ \label{code:return}
        \EndFunction
    \end{algorithmic}
\end{algorithm}

\begin{remark}
    Note that on \cref{line:relaxfootnote} of \cref{alg:find-pivots}, \cref{line:relax2} of \cref{alg:base-case} and \cref{code:relax_condition} of \cref{alg:main}, the conditions are ``$\extd{u} + \edgeweight{u}{v} \leq \extd{v}$''. The equality is required so that an edge relaxed on a lower level can be re-used on upper levels.
\end{remark}

\begin{remark} \label{remark:amortized-time}
    Using methods in \cref{lemma:partition} to implement data structure $\mathcal{D}$ in \cref{alg:main} at level $l$, where $M = 2^{(l-1)t}$, and $|S|\le 2^{lt}$, the total number of insertions $N$ is $O(k2^{lt})$, because of the fact that $|U|=O(k2^{lt})$, the constant-degree property, and the disjointness of $U_i$'s    
    (as established later in \cref{remark:u_i-disjointness}).
    Also, size of $K$ each time is bounded by $O(|U_i|)=O(k2^{(l-1)t})$. 
    Thus, insertion to $\mathcal{D}$ takes $O(\log k + t)=O(t)$ time, and Batch Prepend takes $O(\log k)=O(\log\log n)$ time per vertex in $K$.
\end{remark}


\subsection{Observations and Discussions on the Algorithm}\label{subsec:remark}

We first give some informal explanations on the algorithm, and then in the next subsections, we will formally prove the correctness and running time of the algorithm.

\paragraph{What can we get from the recursion?} Let's look at the recursion tree $\htree$ of~\cref{alg:main} where each node $x$ denotes a call of the \textsc{BMSSP} procedure. Let $\layer_x$ , $B_x$, and $S_x$ denote the parameters $\layer$, $B$, and $S$ at node $x$, $B'_x$ and $U_x$ denote the returned values $B'$ and $U$ at node $x$, $P_x$ and $W_x$ denote $P$ and $W$ (on \cref{code:wp}) returned by the \textsc{FindPivots} procedure at node $x$, respectively, and $W'_x$ denote the set of vertices $v\in W_x$ satisfying $\exttrued{v}<B'_x$. Then we expect:
\begin{enumerate}
    \item Since we assume that all vertices are reachable from $s$, in the root $r$ of $\htree$, we have $U_r=V$;
    \item $|S_x|\leq 2^{\layer_xt}$, so the depth of $\htree$ is at most $(\log n) / t=O(\log^{1/3} n)$; 
    \item We break when $|U_x| \ge k2 ^ {l_xt} > |S_x|$. Intuitively, the size of $U_x$ increases slowly enough so we should still have $|U_x| = O(k2 ^ {l_xt})$ (formally shown in \cref{lemma:size-constraint});
    \item \label{itm:success} If node $x$ is a successful execution, in the execution of \textsc{FindPivots}, $\expectU$ in~\cref{lemma:findpivots} is equal to $U_x$, so $|P_x|\leq |U_x|/k$ by~\cref{lemma:findpivots};
    \item \label{itm:partial} If node $x$ is a partial execution, $|U_x|\geq k2^{\layer_xt}\geq k|S_x|$, so $|P_x|\leq |S_x|\leq |U_x|/k$;
    \item \label{itm:tree} For each node $x$ of $\htree$ and its children $y_1,\cdots,y_q$ in the order of the calls, we have:
    \begin{itemize}
        \item $U_{y_1},\cdots, U_{y_q}$ are disjoint. For $i < j$, distances to vertices in $U_{y_i}$ are smaller than distances to vertices in $U_{y_j}$;
        \item $U_x=W'_x\cup U_{y_1}\cup\cdots\cup U_{y_q}$;
        \item Consequently, the $U_x$'s for all nodes $x$ at the same depth in $\htree$ are disjoint, and total sum of $|U_x|$ for all nodes $x$ in $\htree$ is $O(n(\log n) / t)$. 
    \end{itemize}
\end{enumerate}

\paragraph{Correctness.} In a call of \textsc{BMSSP}, let $\expectU$ denote the set that contains all vertices $v$ with $\exttrued{v} < B$ and the shortest path to $v$ visits some vertex in $S$. 
Then \textsc{BMSSP} should return $\actualU = \expectU$ in a successful execution, or $\actualU = \{u\in \expectU$ : $\exttrued{u} < \actualB\}$ in a partial execution, with all vertices in $U$ complete.

In the base case where $\layer = 0$, $S$ contains only one vertex $x$, a Dijkstra-like algorithm starting from $x$ completes the task.

Otherwise we first shrink $S$ to a smaller set of pivots $P$ to insert into $\ds$, with some vertices complete and added into $W$. \Cref{lemma:findpivots} ensures that the shortest path of any remaining incomplete vertex $v$ visits some complete vertex in $\ds$. For any bound $\expectB_i \leq \expectB$, if $\exttrued{v} < \expectB_i$, the shortest path to $v$ must also visit some complete vertex $u\in \ds$ with $\exttrued{u} < \expectB_i$. Therefore we can call subprocedure \textsc{BMSSP} on $\expectB_i$ and $S_i$.

By inductive hypothesis, each time a recursive call on \cref{code:recursive} of \cref{alg:main} returns, 
vertices in $\actualU_i$ are complete.
After the algorithm relaxes edges from $u \in \actualU_i$ and inserts all the updated  out-neighbors $x$ with $\extd{x} \in [\expectB'_i, \expectB)$ into $\ds$, once again any remaining incomplete vertex $v$ now visits some complete vertex in $\ds$.

Finally, with the complete vertices in $W$ added, $U$ contains all vertices required and all these vertices are complete.

\paragraph{Running time.} The running time is dominated by calls of \textsc{FindPivots} and the overheads inside the data structures $\ds$. By (\ref{itm:success}) and (\ref{itm:partial}), the running time of \textsc{FindPivots} procedure at node $x$ is bounded by $O(|U_x|k)$, so the total running time over all nodes on one depth of $\htree$ is $O(nk)$. Summing over all depths we get $O(nk\cdot (\log n) / t)=O(n\log^{2/3}n)$.

For the data structure $\ds$ in a call of \textsc{BMSSP}, the total running time is dominated by \textsc{Insert} and \textsc{Batch Prepend} operations. 
We analyze the running time as follows.
\begin{itemize}
    \item Vertices in $P$ can be added to $\ds$ on \cref{code:ds_insert}. Since $|P_x|=O(|U_x|/k)$, the total time for nodes of one depth of $\htree$ is $O(n/k\cdot (t+\log k))=O(nt/k)$. Summing over all depths we get $O((n\log n)/k)=O(n\log^{2/3}n)$.
    \item Some vertices already pulled to $S_i$ can be added back to $\ds$ through \textsc{BatchPrepend} on \cref{code:batch_prepend}. Here we know in every iteration $i$, $|S_i|\leq |U_i|$, so the total time for adding back is bounded by $\sum_{x\in\htree}|U_x|\cdot \log k=O(n \log k \cdot (\log n) / t) = O(n\cdot\log^{1/3}n\cdot\log\log n)$.
    \item A vertex $v$ can be inserted to $\ds$ through edge $(u,v)$ \textbf{directly} on \cref{code:single_insert} if $\expectB_i\leq \exttrued{u} + \edgeweight{u}{v} < \expectB$, or \textbf{through $K$} as on \cref{code:k_insert} if $\actualB_i\leq \exttrued{u} + \edgeweight{u}{v} < \expectB_i$. (Note that $u\in \actualU_{i}$ is already complete.) Every edge $(u,v)$ can only be relaxed once in each level by (\ref{itm:tree}), but it can be relaxed in multiple levels in an ancestor-descendant path of the recursion tree $\htree$. (Note that the condition on \cref{code:relax_condition} is $\extd{u} + \edgeweight{u}{v} \leq \extd{v}$.) However, we can show every edge $(u,v)$ can only lead $v$ to be directly inserted to $\ds$ by the \textsc{Insert} operation on one level.
    \begin{itemize}
        \item It can be easily seen that if $y$ is a descendant of $x$ in the recursion tree $\htree$, $B_y\leq B_x$.
        \item If $(u,v)$ leads $v$ to be directly inserted to $\ds$ on \cref{code:single_insert}, then $\exttrued{u} + \edgeweight{u}{v}\geq \expectB_i$. Since $u\in \actualU_{i}$, for every descendant $y$ of the current call in the recursion tree $\htree$ satisfying $u\in U_y$, we have $B_y\leq \expectB_i \leq \exttrued{u} + \edgeweight{u}{v}$, so $v$ will not be added to $\ds$ through $(u,v)$ in any way in lower levels.
    \end{itemize}
    \item Since every edge $(u,v)$ can only lead $v$ to be directly inserted to $\ds$ by the \textsc{Insert} operation once in the algorithm, the total  time is $O(m(\log k+t))=O(m\log^{2/3}n)$. The time for all edges $(u,v)$ leading $v$ to be inserted to $\ds$ through $K$ in one level is $O(m\log k)$, and in total is $O(m\log k\cdot (\log n)/t)=O(m\cdot\log^{1/3}n\cdot\log\log n)$.
\end{itemize}


\subsection{Correctness Analysis}\label{subsec:correctness}

To present the correctness more accurately and compactly, we introduce several notations on the ``shortest path tree''. Due to the uniqueness of shortest paths, the shortest path tree $\rawtree$ rooted at the source $s$ can be constructed unambiguously. 
Define $\tree{u}$ as the subtree rooted at $u$ according to $\exttrued{\cdot}$. An immediate observation is that \emph{the shortest path to $v$ passes through $u$ if and only if $v$ is in the subtree rooted at $u$}.

For a vertex set $S$, define $\tree{S} = \bigcup_{v\in S}\tree{v}$, namely the union of the subtrees of $T(s)$ rooted at vertices in $S$, or equivalently, the set of vertices whose shortest paths pass through some vertex in $S$.
For a set of vertices $S\subseteq V$, denote $S^* = \{v\in S: v \text{ is complete}\}$. Then $\tree{S^*} = \bigcup_{v\in S^*}\tree{v}$ is the union of subtrees of $T(s)$ rooted at $S^*$, or the set of vertices whose shortest paths pass through some complete vertex in $S$. Obviously, $\tree{S^*} \subseteq \tree{S}$.
Note that $\tree{S^*}$ is sensitive
to the progress of the algorithm while $\tree{S}$ is a fixed set. Another observation is that, a complete vertex would remain complete, so $S^*$ and $\tree{S^*}$ never lose a member.

For a bound $B$, denote $\boundtree{S}{B} = \{v\in \tree{S}: \exttrued{v} < B\}$. Also note that $\boundtree{S}{B}$ coincides with $\expectU$ mentioned above. For an interval $[b, B)$, denote $\rangetree{S}{[b, B)} =\{v\in \tree{S}: \exttrued{v} \in [b, B)\}$.



\begin{lemma}[Pull Minimum]\label{lemma:pull-minimum}
Suppose every incomplete vertex $v$ with $\exttrued{v} < B$ is in $\tree{S^*}$. Suppose we split $S$ into $X=\{x\in S:\extd{x}<\mathcal{B}\}$ and $Y=\{x\in S:\extd{x}\geq \mathcal{B}\}$ for some $\mathcal{B}<B$.
Then every incomplete vertex $v$ with $\exttrued{v} < \mathcal{B}$ is in $\tree{X^*}$. Moreover, for any $\mathcal{B}' < \mathcal{B}$, $\boundtree{S}{\mathcal{B}'} = \boundtree{X}{\mathcal{B}'}$.
\end{lemma}

\begin{proof}
For any incomplete vertex $v$ with $\exttrued{v} < \mathcal{B} $, as $\mathcal{B} < B$, by definition, $v\in \tree{u}$ for some complete vertex $u\in S$. Therefore $\extd{u} = \exttrued{u} \leq \exttrued{v} < \mathcal{B}$, so $u\in X$ and thus $v \in \tree{X^*}$.

For the second statement, it is clear that $\boundtree{X}{\mathcal{B}'} \subseteq \boundtree{S}{\mathcal{B}'}$. For any $v\in \boundtree{S}{\mathcal{B}'}$, since $d(v) < \mathcal{B}' < \mathcal{B} < B$, the shortest path to $v$ passes through some vertex $x\in S^*$. Now that $\extd{x} = \exttrued{x} \leq \exttrued{v} < \mathcal{B}'$, we have $x\in X$ and $v\in \boundtree{X}{\mathcal{B}'}$.
\end{proof}

Now we are ready to prove the correctness part of \cref{lemma:bmssp}.

\begin{lemma}\label{lemma:main-algo-correctness}
We prove the correctness of \cref{alg:main} by proving the following statement (the size of $\actualU$ is dealt with in \cref{lemma:size-constraint}), restated from \cref{lemma:bmssp}:
Given a level $\layer \in [0, \lceil (\log n) / t\rceil ]$, a bound $B$ and a set of vertices $S$ of size $\leq 2^{\layer t}$, suppose that every incomplete vertex $v$ with $\exttrued{v} < \expectB$ is in $\tree{S^*}$.

Then, after running \cref{alg:main}, we have: $\actualU = \boundtree{S}{\actualB}$ , and $\actualU$ is complete.
\end{lemma}
\begin{proof}
We prove by induction on $\layer$. When $\layer =  0$, since $S=\{x\}$ is a singleton, $x$ must be complete. Then clearly the classical Dijkstra's algorithm (\cref{alg:base-case}) finds the desired $\actualU = \boundtree{S}{\actualB}$ as required.
Suppose correctness holds for $\layer - 1$, we prove that it holds for $\layer$.
Denote $\dsiter{i}$ as the set of vertices (keys) in $\ds$ just before the $i$-th iteration (at the end of the $(i-1)$-th iteration), and $\dsiter{i}^*$ as the set of complete vertices in $\dsiter{i}$ at that time.
Next, we prove the following two propositions by induction on increasing order of $i$: 



Immediately before the $i$-th iteration,
\begin{itemize}
    \item[(a)] Every incomplete vertex $v$ with $\exttrued{v} < B$ is in $\rangetree{P}{[\actualB_{i - 1}, \expectB)}$.
    \item[(b)] $\rangetree{P}{[\actualB_{i - 1}, \expectB)} = \boundtree{\dsiter{i}}{\expectB} = \boundtree{\dsiter{i}^*}{\expectB}$; 
\end{itemize}

In the base case $i=1$, by \cref{lemma:findpivots}, for every incomplete vertex $v$ with $\exttrued{v} < \expectB$ (including $v\in \pivotP$), $v \in \boundtree{\pivotP^*}{\expectB}$. Then $\exttrued{v}\geq \min_{x\in \pivotP^*}\exttrued{x} = \min_{x\in \pivotP^*}\extd{x} \geq \actualB_0$. Thus every such $v$ is in $\rangetree{\pivotP}{[\actualB_0, \expectB)}$ (actually $\boundtree{\pivotP}{\actualB_0}$ is an empty set) 
and $\rangetree{\pivotP}{[\actualB_0, \expectB)} = \boundtree{\pivotP}{\expectB} = \boundtree{\pivotP^*}{\expectB}$. Because $\dsiter{1} = \pivotP$, the base case is proved.

Suppose both propositions hold for $i$. Then each incomplete vertex $v$ with $\exttrued{v} < B$ is in $\rangetree{P}{[\actualB_{i - 1}, \expectB)} \subseteq \tree{\dsiter{i}^*}$. By \cref{lemma:partition}, the suppositions of \cref{lemma:pull-minimum} are met for $X := S_i$, $Y := \ds_{i} \backslash S_i$, and $\mathcal{B} := \expectB_i$, so every incomplete vertex $v$ with $d(v) < \expectB_i$ is in $T({S_i} ^ *)$; also note $|S_i|\leq 2^{(l-1)t}$.
By induction hypothesis on level $\layer - 1$, the $i$-th recursive call is correct and we have $\actualU_i = \boundtree{S_i}{\actualB_i}$ and is complete.
Note that $S_i$ includes all the vertices $v$ in $\dsiter{i}$ such that $\extd{v}<\actualB_{i}$ and that $\actualB_{i} \le B$. Thus by \Cref{lemma:pull-minimum} and proposition (b) on case $i$, $\actualU_i = \boundtree{S_i}{\actualB_i} = \boundtree{\dsiter{i}}{\actualB_i} = \rangetree{P}{[\actualB_{i - 1}, \actualB_{i})}$ and is complete (thus proving \Cref{remark:u_i-disjointness}).
Then, every incomplete vertex $v$ with $\exttrued{v} < B$ is in $\rangetree{\pivotP}{[\actualB_{i-1}, \expectB)} \setminus \rangetree{\pivotP}{[\actualB_{i-1}, \actualB_{i})} = \rangetree{\pivotP}{[\actualB_{i}, \expectB)}$. Thus, we proved proposition (a) for the $(i+1)$-th iteration.

Now we prove proposition (b) for the $(i+1)$-th iteration. Since $\actualB_{i} \geq \actualB_{i - 1}$, from proposition (b) of the $i$-th iteration, we have $\rangetree{P}{[\actualB_{i}, \expectB)} = \rangetree{\dsiter{i}}{[\actualB_{i}, \expectB)} = \rangetree{\dsiter{i}^*}{[\actualB_{i}, \expectB)}$.
Suppose we can show that $\rangetree{\dsiter{i}^*}{[\actualB_{i}, \expectB)} \subseteq \boundtree{\dsiter{i + 1}^*}{\expectB}$ and $\boundtree{\dsiter{i+1}}{\expectB} \subseteq \rangetree{\dsiter{i}}{[\actualB_{i}, \expectB)}$. By definition $\boundtree{\dsiter{i + 1}^*}{\expectB} \subseteq \boundtree{\dsiter{i + 1}}{\expectB}$, then
\[\rangetree{\pivotP}{[\actualB_i, \expectB)} = \rangetree{\dsiter{i}^*}{[\actualB_i, \expectB)} \subseteq \boundtree{\dsiter{i + 1}^*}{\expectB} \subseteq \boundtree{\dsiter{i + 1}}{\expectB} \subseteq \rangetree{\dsiter{i}}{[\actualB_i, \expectB)} = \rangetree{\dsiter{i}^*}{[\actualB_i, \expectB)},\]
and by sandwiching, proposition (b) holds for $(i + 1)$. Thus it suffices to prove $\rangetree{\dsiter{i}^*}{[\actualB_{i}, \expectB)} \subseteq \boundtree{\dsiter{i + 1}^*}{\expectB}$ and $\boundtree{\dsiter{i + 1}}{\expectB} \subseteq \rangetree{\dsiter{i}}{[\actualB_{i}, \expectB)}$.

For any vertex $y\in \dsiter{i}^* \setminus \dsiter{i + 1}^*$, we have $y \in S_i$. Note that vertices $v$ of $S_i$ with $\extd{v}\geq \actualB_i$ are batch-prepended back to $\dsiter{i + 1}$ on \cref{code:batch_prepend}, so we have $\extd{y} < \actualB_i$. Thus $y\in \actualU_i$ and is complete. For any vertex $x\in \rangetree{y}{[\actualB_i, \expectB)}$, as $x\not \in \actualU_i$, along the shortest path from $y$ to $x$, there is an edge $(u, v)$ with $u\in \actualU_i$ and $v\not\in \actualU_i$. During relaxation, $v$ is then complete and is added to $\dsiter{i+1}$, so $x\in \rangetree{v}{[\actualB_i, \expectB)} \subseteq \boundtree{\dsiter{i + 1}^*}{\expectB}$.

For any vertex $v\in \dsiter{i + 1} \setminus \dsiter{i}$, since $v$ is added to $\dsiter{i + 1}$, 
there is an edge $(u, v)$ with $u\in \actualU_i = \boundtree{\dsiter{i}}{\actualB_{i}}$ and the relaxation of $(u, v)$ is valid ($\extd{u} + \edgeweight{u}{v} \leq \extd{v}$). Pick the last valid relaxation edge $(u,v)$ for $v$. If $v\in \tree{u}$, then $v$ is complete, $\exttrued{v}=\extd{v}\geq \actualB_{i}$, and $v\in \rangetree{\dsiter{i}}{[\actualB_{i}, B)}$; if $v$ is incomplete, then by proposition (a), $v\in \rangetree{\pivotP}{[\actualB_i, \expectB)} = \rangetree{\dsiter{i}}{[\actualB_i, \expectB)}$; or if $v$ is complete but $v \notin \tree{u}$, we have a contradiction because it implies that relaxation of $(u, v)$ is invalid.

Now that we have proven the propositions, we proceed to prove that the returned $\actualU = \boundtree{S}{\actualB}$ and that $\actualU$ is complete. Suppose there are $q$ iterations.
We have shown that every $U_i = \rangetree{P}{[\actualB_{i - 1}, \actualB_i)}$ and is complete, so $U_i$'s are also disjoint. By \cref{lemma:findpivots}, $\pivotU$ contains every vertex in $\boundtree{S \setminus \pivotP}{B}$ (as they are not in $\tree{\pivotP}$) and they are complete; besides, all vertices $v$ in $W$ but not in $\boundtree{S \setminus \pivotP}{B}$ with $\exttrued{v} < \actualB_q$ are complete (because they are in $\boundtree{\pivotP}{\actualB_q}$). Therefore, $\{x\in \pivotU: \extd{x} < \actualB_q\}$ contains every vertex in $\boundtree{S \setminus \pivotP}{\actualB_q}$ and is complete.
Thus finally $\actualU := (\bigcup_{i = 1}^{q}U_i)\cup \{x \in \pivotU: \extd{x} < \actualB_{q}\}$ equals $\boundtree{S}{\actualB_{q}}$ and is complete.
\end{proof}

\begin{remark}\label{remark:u_i-disjointness}
From the proof of \cref{lemma:main-algo-correctness}, $\actualU_i = \rangetree{\pivotP}{[\actualB_{i - 1}, \actualB_{i})}$ and they are disjoint and complete. As a reminder, \cref{lemma:pull-minimum} and proposition (b) of \cref{lemma:main-algo-correctness} prove this.
\end{remark}

\subsection{Time Complexity Analysis}

\begin{lemma}
\label{lemma:size-constraint}
Under the same conditions as in \cref{lemma:main-algo-correctness}, where every incomplete vertex $v$ with $\exttrued{v} < B$ is in $\boundtree{S^*}{\expectB}$.
After running \cref{alg:main}, we have $|\actualU| \leq 4k2^{\layer t}$. If $\actualB < \expectB$, then $|\actualU| \geq k2^{lt}$.
\end{lemma}
\begin{proof}
We can prove this by induction on the number of levels. The base case $\layer = 0$ clearly holds.

Suppose there are $q$ iterations. Before the $q$-th iteration, $|\actualU| < k2^{\layer t}$. After the $q$-th iteration, the number of newly added vertices $|\actualU_q| \leq 4k2^{(\layer - 1) t}$ (by \cref{lemma:main-algo-correctness}, the assumptions for recursive calls are satisfied), and because $|\pivotU| \leq k|S| \leq k2^{\layer t}$, we have $|\actualU| \leq 4k2^{\layer t}$.
If $\ds$ is empty, the algorithm succeeds and quits with $\actualB = \expectB$. Otherwise, the algorithm ends partially because $|\actualU| \geq k2^{\layer t}$.
\end{proof}

\begin{lemma}\label{lemma:make-progress}
Immediately before the $i$-th iteration of \cref{alg:main}, $\min_{x\in \ds} \exttrued{x} \geq \actualB_{i-1}$.
\end{lemma}

\begin{proof}
From the construction of $\ds$, immediately before the $i$-th iteration of \cref{alg:main}, $\min_{v\in \ds}\extd{v} \geq \actualB_{i - 1}$. For $v\in \ds$, if $v$ is complete, then $\exttrued{v} = \extd{v} \geq \actualB_{i - 1}$. If $v$ is incomplete, by proposition (a) of \cref{lemma:main-algo-correctness}, because $v\in \rangetree{\pivotP}{[\actualB_{i - 1}, \expectB)}$, $\exttrued{v} \geq \actualB_{i - 1}$.
\end{proof}

\begin{lemma}\label{lemma:pivot-size}
Denote by $\expectU := \boundtree{S}{B}$ the set that contains every vertex $v$ with $\exttrued{v} < \expectB$ and the shortest path to $v$ visits some vertex of $S$.
After running \cref{alg:main}, $\min\{|\expectU|, k\abs{S}\} \leq |\actualU|$ and $|S| \leq |U|$.
\end{lemma}
\begin{proof}
If it is a successful execution, then $S\subseteq \expectU = \actualU$.
If it is partial, then $k\abs{S} \leq k2^{\layer t} \leq \abs{\actualU}$.
\end{proof}

For a vertex set $U$ and two bounds $c < d$, denote $N^{+}(U) = \{(u, v): u\in U\}$ as the set of outgoing edges from $U$, and $N^{+}_{[c, d)}(U) = \{(u, v): u\in U \text{ and } \exttrued{u} + \edgeweight{u}{v} \in [c, d)\}$.

\begin{lemma}[Time Complexity]\label{lemma:main-algo-time}
Suppose a large constant $\timeconst$ upper-bounds the sum of the constants hidden in the big-O notations in all previously mentioned running times: in find-pivots, in the data structure, in relaxation, and all other operations. \cref{alg:main} solves the problem in \cref{lemma:bmssp} in time
\[\timeconst(k + 2 t / k)(\layer + 1)|\actualU| + \timeconst (t + l \log k) |N^{+}_{[\min_{x\in S}\exttrued{x}, \expectB)}(\actualU)|.\]
With $k = \lfloor \log^{1/3}(n) \rfloor$ and $t = \lfloor \log^{2/3}(n) \rfloor$, calling the algorithm with $l = \lceil (\log n) / t\rceil, S = \{s\}, B = \infty$ takes $O(m\log^{2/3}n)$ time.

\end{lemma}
\begin{proof}

We prove by induction on the level $\layer$.
When $\layer =  0$, the algorithm degenerates to the classical Dijkstra's algorithm (\cref{alg:base-case}) and takes time $C|\actualU|\log k$, so the base case is proved.
Suppose the time complexity is correct for level $\layer - 1$. Now we analyze the time complexity of level $\layer$.

By \cref{remark:amortized-time}, each insertion takes amortized time $\timeconst t$ (note that $\log k < t$); each batch prepended element takes amortized time $\timeconst\log k$; each pulled term takes amortized constant time. But we may ignore \textsc{Pull}'s running time because each pulled term must have been inserted/batch prepended, and the constant of \textsc{Pull} can be covered by $\timeconst$ in those two operations.


By \cref{remark:u_i-disjointness}, $\actualU_i$'s are disjoint and $\sum_{i\geq 1}|\actualU_i| \leq |\actualU|$.

Now we are ready to calculate total running time of \cref{alg:main} on level $l$ by listing all the steps.

By \cref{lemma:findpivots}, the \textsc{FindPivots} step takes time $\findpivotconst\cdot\min\{|\expectU|, k|S|\}k$ and $|\pivotP| \leq \min\{|\expectU| / k, |S|\}$. Inserting $\pivotP$ into $\ds$ takes time $\partitionconst |\pivotP|t$. And by \cref{lemma:pivot-size}, their time is bounded by $\timeconst (k + t / k)\abs{\actualU}$.


In the $i$-th iteration, the sub-routine takes
$\timeconst (k + 2 t / k)\layer|\actualU_i| + \partitionconst (t + (l-1)\log k)|N^{+}_{[\min_{x\in S_i}\exttrued{x}, \expectB_{i})}(\actualU_i)|$
time by the induction hypothesis.
Taking the sum, by \cref{lemma:make-progress}, $N^{+}_{[\min_{x\in S_i}\exttrued{x}, \expectB_i)}(\actualU_i) \subseteq N^{+}_{[\actualB_{i-1}, \expectB_i)}(\actualU_i)$. Then the sum of time spent by all the sub-routines is bounded by
\[\timeconst (k + 2 t / k)\layer |\actualU| + \overbrace{\timeconst(t + (\layer - 1) \log k)\sum_{i \geq 1}|N^{+}_{[\actualB_{i-1}, \expectB_{i})}(\actualU_i)|}^{\mathcal{N}_1}.\]

In the following relaxation step,
for each edge $(u, v)$ originated from $\actualU_{i}$, if $\extd{u} + \edgeweight{u}{v} \in [\expectB_i, \expectB)$, they are directly inserted; if $\extd{u} + \edgeweight{u}{v} \in [\actualB_{i}, \expectB_{i})$, they are batch prepended. Again by \cref{remark:u_i-disjointness}, all $\actualU_i$'s are complete, so for the previous statements we can replace $\extd{\cdot}$ with $\exttrued{\cdot}$. Thus in total, it takes time
\[\overbrace{\partitionconst t\sum_{i\geq 1}|N^{+}_{[\expectB_{i}, \expectB)}(\actualU_{i})|}^{\mathcal{N}_2} + \overbrace{\partitionconst (\log k) \sum_{i\geq 1} |N^{+}_{[\actualB_{i}, \expectB_{i})}(\actualU_{i})|}^{\mathcal{N}_3}.\]

Vertices in $S_{i}$ were also batch prepended, and this step takes time $O(|\{x\in S_i: \extd{x}\in [\actualB_i, \expectB_i)\}|\log k)$. This operation takes non-zero time only if $\actualB_i < \expectB_i$, i.e., the $i$-th sub-routine is a partial execution. Thus in total this step takes time $\partitionconst \sum_{\text{partial } i}|S_i|\log k \leq (\partitionconst |\actualU|\log k) /k \leq \partitionconst |\actualU| t /k$.

Therefore, the total running time on the level $\layer$ is
\[\timeconst(k + 2 t / k) (\layer + 1) |\actualU| + \mathcal{N},\]
where $\mathcal{N}$ is the sum
\[\mathcal{N} = \overbrace{\timeconst (t +  (\layer - 1) \log k)\sum_{i\geq 1}|N^{+}_{[\actualB_{i-1}, \expectB_{i})}(\actualU_i)| + \timeconst t \sum_{i\geq 1}|N^{+}_{[\expectB_{i}, \expectB)}(\actualU_i)|}^{\mathcal{N}_1+\mathcal{N}_2} + \overbrace{\timeconst (\log k )\sum_{i\geq 1}|N^{+}_{[\actualB_{i}, \expectB_i)}(\actualU_i)|}^{\mathcal{N}_3}.\]

For $\mathcal{N}_1+\mathcal{N}_2$, because
\[\sum_{i\geq 1}|N^{+}_{[\actualB_{i-1}, \expectB_{i})}(\actualU_{i})| + \sum_{i\geq 1}|N^{+}_{[\expectB_i, \expectB)}(\actualU_{i})| = \sum_{i\geq 1}|N^{+}_{[\actualB_{i-1}, \expectB)}(\actualU_{i})| \leq |N^{+}_{[\actualB_0, \expectB)}(\actualU)|, \]
and $\actualB_0\geq \min_{x\in S}\exttrued{x}$, it is bounded by $\timeconst (t + (\layer - 1)\log k)|N^{+}_{[\min_{x\in S}\exttrued{x}, \expectB)}(\actualU)|$.
$\mathcal{N}_3$ is bounded by $\timeconst (\log k) |N^{+}_{[\actualB_{0}, \expectB)}(\actualU)| \leq \timeconst (\log k) |N^{+}_{[\min_{x\in S}\exttrued{x}, \expectB)}(\actualU)|$.

Therefore $\mathcal{N}\leq \timeconst (t + \layer \log k) |N^{+}_{[\min_{x\in S}\exttrued{x}, \expectB)}(\actualU)| $.
\end{proof}

\bibliographystyle{alpha}
\bibliography{references}

\end{document}